%% file: overview-sn-modeling.tex
\documentclass[12pt,reqno]{article}
\pdfoutput=1

\usepackage{graphicx}
\usepackage{ifpdf}

\usepackage{hyperref}
\makeatletter
\def\@fnsymbol#1{\ensuremath{\ifcase#1\or *\or \sharp\or \mathchar "278\or
    \mathchar "27B\or \| \or \dagger\or \ddagger\or \dagger\dagger \or
    \ddagger\ddagger \else\@ctrerr\fi}}
\makeatother

\usepackage[latin1]{inputenc}
\usepackage{url}
\usepackage{amsthm}
\usepackage{amsmath}
\usepackage{marvosym}
\usepackage{txfonts}
\usepackage{colortbl}
\usepackage{booktabs}

\newcommand{\F}{\mathcal{F}}
\DeclareMathOperator{\conv}{conv}
\newcommand{\maximal}{\mbox{maximal}}
\newcommand{\m}[1]{\ensuremath\mathbf{#1}} 
\newcommand{\ve}[1]{\ensuremath\mathbf{#1}} 

\def\lor{\vee}
\def\land{\wedge}
\def\liff{\leftrightarrow}

\newenvironment{pmat}{%
  \left(%
    \begin{array}{*{10}r}
    }{%
    \end{array}%
  \right)%
}

\def\CPLEX{{\small CPLEX}}
\def\SUN{{\small SUN}}
\def\OR{{\small OR}}
\def\AND{{\small AND}}
\def\SAT{{\small SAT}}
\def\TRUE{{\textsc{True}}}
\def\FALSE{{\textsc{False}}}
\def\CNF{{\small CNF}}
\def\IFFSAT{{\small IFFSAT}}
\def\IFF{{\small IFF}}
\def\NP{{\ensuremath{\mathcal NP}}}

\def\RIFFSAT{{\small RIFFSAT}}
\def\TCR{{\small TCR}}

\theoremstyle{plain}      \newtheorem{theorem}{Theorem}
\theoremstyle{plain}      \newtheorem{lemma}[theorem]{Lemma}
\theoremstyle{plain}      
\theoremstyle{plain}      
\theoremstyle{plain}      
\theoremstyle{definition} \newtheorem{definition}{Definition}
\theoremstyle{remark}     \newtheorem{remark}{Remark}
\theoremstyle{remark}     \newtheorem{example}{Example}
\theoremstyle{remark}     \newtheorem{problem}{Problem}

\title{Logic Integer Programming\\ Models for Signaling Networks}
\author{%
\setcounter{footnote}{0}
Utz-Uwe Haus\thanks{Institut f\"ur Mathematische Optimierung, 
Otto-von-Guericke-Universit\"at Magdeburg,  Universit\"atsplatz 2, D-39106
Magdeburg, Germany, phone: +49 391 6718646, fax: +49 391 6711171,
e-mail: \{haus,niermann,weismantel\}@imo.math.uni-magdeburg.de} \and
\setcounter{footnote}{0}
Kathrin Niermann\footnotemark \and 
Klaus Truemper\thanks{Department of
    Computer Science, University of Texas at Dallas, Richardson,
    Texas 75080, USA, phone: +1 972 883-2712, email:truemper@utdallas.edu}
  \and
\setcounter{footnote}{0}
Robert Weismantel\footnotemark}

\date{}
\usepackage{natbib}

\begin{document}


\maketitle

\begin{abstract}
  We propose a static and a dynamic approach to model
  biological signaling networks, and show how each can be used to
  answer relevant biological questions. For this we use the two
  different mathematical tools of Propositional Logic and Integer
  Programming. The power of discrete mathematics for handling
  qualitative as well as quantitative data has so far not been
  exploited in Molecular Biology, which is mostly driven by
  experimental research, relying on first-order or statistical models.
  The arising logic statements and integer programs are analyzed and
  can be solved with standard software. For a restricted class of
  problems the logic models reduce to a polynomial-time solvable
  satisfiability algorithm. Additionally, a more dynamic
  model enables enumeration of possible time resolutions in
  poly-logarithmic time. Computational experiments are included.\\ 
  

\end{abstract}


\section{Introduction}

Cellular decisions are determined by highly complex molecular
interactions. In some biological systems like \emph{Saccharomyces} or
\emph{E. coli},
detailed measurements of the interacting molecules, including reaction
kinetics, have successfully been performed, allowing the construction
of quantitative models~\citep{feist:07}. In many other systems such extensive
measurements are not available, because of practical experimental
restrictions or ethical constraints. However, in these cases, there is
often still a sizable amount of qualitative information available, but
a lack of suitable predictive modeling tools.

We focus here on interactions in form of signal transduction
processes. For such a process we assume that a set of molecules that
are important for the biological unit is known. The biological unit
reacts to external signals or environmental challenges like
stimulation or infection. Typically, the molecules may be subdivided
into input components (e.g. receptors), intermediate components and
output components (e.g. transcription factors): When an external
signal arrives, this signal is processed through the entire unit by
first influencing a subset of the input components. Activation state
or presence/absence information is propagated through intermediate nodes
towards some of the output molecules. Based on the assumption that we
know the ``local'' mechanism of activation, it is our goal to predict
the global behavior of the system, identifying the underlying network
structure. For our purposes `activation' can mean any interaction that
can be explained biologically: A protein may be considered activated
after phosphorylation, Ca$^{++}$ flux may be detected, or a component
may change its location within the biological unit. Similarly,
`biological unit' need not be restricted to a single cell or
compartment, but any collection of components that are to be
considered.

Subsequently we propose a logic and an integer programming model to
analyze the static behavior of signaling networks. With both
approaches one is able to verify the biological modeling, find
potential failure modes and determine suitable intervention
strategies. But some phenomena in biological units like time delays
and (negative) feedback loops can not be modeled in this static
fashion. Thus we focus on the dynamics of signaling networks in the
second part. We extend the previous models so that activations can be
modeled as occurring at different time points or with different
signaling speed. In this section we present a generic framework for
computations with the dynamic model which can use solvers specific to
each of the modeling techniques, logic and integer programming, as
oracles. This ability requires a basic understanding of the
transformation from satisfiability systems to integer programming
models and vice versa. Therefore we discuss both sides in the static
as well as in the dynamic context hand in hand. The last section
provides computational tests showing the appropriateness of our
techniques.

\newpage

\section{Modeling Logical Interactions}
\label{sec:logical-models}

The simplest model of signaling processes is to 
collect local data in the form of logical formulas, that can be
written down in propositional logic~\citep{immu-2}: Introduce logical
variables for each component under consideration, and write down
implication formulas for experimentally proven knowledge statements
like \emph{``MEK activates ERK''} as
$$\mbox{MEK} \to \mbox{ERK}$$
and \emph{``In the absence of (activated) pten and ship1 we find that pi3k
generates (active) pip3''} as
$$\lnot{\mbox{pten}}\wedge\lnot{\mbox{ship1}}\wedge \mbox{pi3k} \to
\mbox{pip3}.$$

Let the formulas be denoted as $S_i$ with
$i\in\left\{1,\dots,s\right\}$. We can then identify the formula
$S=\bigwedge_{i=1}^s S_i$
with the model of the biological unit considered: All logical
statements $S_i$ should be valid at the same time to
model the global behavior of the unit. We will, as usual, use $A\to B$
as abbreviation of $(\lnot A)\lor B$, and $A \leftrightarrow B$
instead of $(A\to B)\land(B\to A)$. For the reader unfamiliar with
the formalism of propositional logic we refer
to~\citep{buening-lettmann:prop-logic}.

We will henceforth assume  that all implications in
the set $S_i$ are given in \OR-form $\bigvee_{j\in L^A_i} A_j \to
\bigvee_{j\in L^B_i} B_j$ for literals $A_j$, $B_j$. Here, $L^A_i$ is
the set of literals appearing on the left in formula $i$. We will also
require that the set $S$ has been extended with reverse implications
by first aggregating formulas with the same right-hand-side, and then
adding the reverse implication. Thus, for each set of literals
$R=\left\{B_j|j\in L^B_i\mbox{ for
    some }i\in\left\{1,\dots,s\right\}\right\}$ appearing on the
right-hand-side of the implications indexed by $I=\left\{i| S_i =
  \left(\bigvee_{j\in L^A_i} A_j \to \bigvee_{r\in R} r\right)\right\}$, we also find
the implication $\bigvee_{i\in I}\bigvee_{j\in L^A_i} A_j \leftarrow
\bigvee_{r\in R} r$ in $S$. This enforces that every activation must have a
`cause' within the given model.

The question whether there exists a pattern of activations
satisfying all formulas of $S$ is an instance of the satisfiability
(\SAT) problem~\citep{truemper:04}. In general, this problem is
to ask for a \emph{truth assignment} such that the logical formula in
\CNF{} form is \TRUE:

\begin{definition}[\CNF{}, truth assignment, satisfiable]\hspace{2cm}
  \begin{enumerate}
  \item A \emph{clause} $\alpha$ is a disjunction of literals,
    i.e. $\alpha=A_1\vee\ldots\vee A_n$ with literals $A_i$.
  \item A formula $\alpha$ is in \emph{Conjunctive Normal Form} (\CNF)
    if and only if $\alpha$ is a conjunction of clauses.
  \item A formula $\alpha$ is in \emph{$k$-Conjunctive Normal Form}
    ($k$-\CNF) if and only if $\alpha$ is a conjunction of $k$-clauses,
    i.e. every clause consists of at most $k$ literals.
  \item A \emph{truth assignment} $\Gamma$ of a propositional formula
    $\alpha$ is defined by
    \begin{align*}
      \Gamma :\left\{\alpha\left|\alpha \mbox{ is a propositional
            formula}\right.\right\} \rightarrow \left\{0,1\right\}.
    \end{align*}
    It can be calculated according to three rules:
    \begin{itemize}
    \item[(i)] $\Gamma (\neg\alpha):=1$, iff $\Gamma(\alpha)=0$
    \item[(ii)] $\Gamma (\alpha\vee\beta):=1$, iff $\Gamma(\alpha)=1$ or $
      \Gamma(\beta)=1$
    \item[(iii)] $\Gamma (\alpha\wedge\beta):=1$, iff $\Gamma(\alpha)=1$
      and $\Gamma(\beta)=1$
    \end{itemize}
  \item A propositional formula $\alpha$ is \emph{satisfiable} if and
    only if there exists a truth assignment $\Gamma$, so that $\Gamma
    (\alpha)=1$.  
  \end{enumerate}
\end{definition}

\begin{definition}\label{def:model}
  A \emph{model m} of a \CNF{} formula $C$ is a satisfying
  truth assignment of $C$. The set of all models of $C$ are denoted
  by $\mbox{models}(C)$. We denote that $m$ assigns 1 (0) to
  variable x by $m(x)=1$ $(m(x)=0)$. If $m(x)=1$ implies $m^*(x)=1$
  for two models $m$ and $m^*$, we say that $m\leq m^*$. If neither
  $m\leq m^*$ nor $m\geq m^*$ is true, the two models are
  \emph{incomparable}. We call a model \emph{maximal} (\emph{minimal})
  if there is no model $m^*$ such that $m<m^*$ ($m>m^*$). We denote
  the set of all maximal (minimal) models of $C$ by $\maximal(C)$
  ($\mbox{minimal}(C)$).
\end{definition}

\begin{problem}[\SAT]
  Given a \CNF{} (3-\CNF) formula $C$, the \emph{\SAT} (\emph{3-\SAT})
  problem is to decide if $C$ is satisfiable and, if so, to return a
  possible model.
\end{problem}

In the setting of \SAT{} problems we can also answer the question
whether, given a partial set of activations, there exists a solution
for the entire formula $S$, by fixing some logical
variables in $S$ to the prescribed values and solving the \SAT{}
problem for the remaining formula $S'$. We are thus prepared to
introduce the signaling network satisfiability problem (\IFFSAT) as:

\begin{problem}[\IFFSAT]
  Let 
  $$\textstyle
  S=\left\{S_i=\bigvee_{j\in L^A_i} A_j \leftrightarrow \bigvee_{j\in
      L^B_i} B_j\right\}$$
  be a set of $s=|S|$ equivalence formulas
  over the literal set $L$, and $L_0,L_1\subseteq L$ two sets of
  variables to be fixed. An instance of the \emph{\IFFSAT} problem is of
  the form
  \begin{equation}  
    \label{eq:snsat}
    \bigwedge_{i=1}^s S_i \wedge \bigwedge_{x\in L_0} (\lnot x)\wedge
    \bigwedge_{x\in L_1} x.\tag{\IFFSAT}
  \end{equation}
\end{problem}

Much research has been done to find effective solution algorithms for
subclasses of \\\SAT~\citep{truemper:04}. There is, however, no algorithm
that can efficiently~\citep{GarJohn79} check satisfiability for
arbitrary propositional formulas, as we have to consider for this
application:
\begin{lemma}
  The satisfiability problem for problems of the form~\eqref{eq:snsat}
  is equivalent to 3-\SAT, hence \NP-complete.
\end{lemma}

\begin{proof}
  We only need to show that 3-\SAT{} instances can be written in \IFFSAT{}
  form. Using $\approx$ to designate logical equivalence this can be
  seen as follows:
{\small
\begin{align*}\textstyle
  &\bigwedge_{j\in \mathcal{I}}x^j_{1}\vee x^j_{2}\vee x^j_{3}\\
  \approx& \bigwedge_{j\in \mathcal{I}}\left[(x^j_{1}\vee x^j_{2}\leftrightarrow u^j_{1})\wedge (u^j_{1}\vee x^j_{3})\right]\\
  \approx& \bigwedge_{j\in \mathcal{I}}(x^j_{1}\vee
  x^j_{2}\leftrightarrow u^j_{1})
         \wedge\bigwedge_{j\in \mathcal{I}} (u^j_{1}\vee
         x^j_{3}\leftrightarrow u^j_{2})\wedge
          \bigwedge_{j\in \mathcal{I}}u^j_{2}
\end{align*}}
This is an instance of \IFFSAT{} form in which at most
$2\left|\mathcal{I}\right|$ literals and $2\left|\mathcal{I}\right|$
propositional formulas were added.
\end{proof}

As already known to Dantzig~\citep{dantzig:63}, \SAT{} problems can be
formulated as integer programs, i.e. feasibility or optimization
problems over linear systems of inequalities, where the solutions are
required to be integral. For an overview of this field
see~\citep{bertsimas-weismantel:ip-book}. For the \IFFSAT{} problem the
associated integer program (IP) is constructed by introducing $|L|$
binary variables $x_l$ and their complements $\bar{x}_l$, and
translating each \IFF{} formula into the system
\begin{equation}\label{eq:sn-sat-ip}
  \begin{array}{rl@{\qquad}r}
    \sum_{j\in L^A_i}  x_{A_j} -  x_{B_k}
    &\geq 0&\text{for all $S_i\in S, k\in L_i^B$}\\
    - x_{A_j} + \sum_{k\in L^B_i}  x_{B_k}
    &\geq 0&\text{for all $S_i\in S, j\in L_i^A$}\\
    x_{l} + \bar{x}_l &= 1    &\mbox{$l\in L$}\\
    x_{p} = 1, x_{q} &= 0&
    \ \mbox{$p\in L_1$},\mbox{$q\in L_0$}
  \end{array}
\end{equation}
where we will assume that for non-negated literals $A\in L$ the
variable $x_A$, and for negated literals $\lnot A\in L$, $\bar{x}_A$
has been used in the formulation of the inequalities.

\begin{remark}
  The integer programming formulation~\eqref{eq:sn-sat-ip}
  of~\eqref{eq:snsat} has the form of a \emph{generalized set cover}
  problem
  \begin{equation}
    \label{eq:gsc}
    Ax\geq 1-n(A), \quad x\in\left\{0,1\right\},
  \end{equation}
  where $n(A)$ is the number of negative entries in the corresponding
  row of $A$.
\end{remark}

We illustrate the presented methods with the help of a small example.

\begin{figure}
  \centering
  \hfil%
  \parbox{.55\textwidth}{
    \resizebox{!}{.5\textwidth}{\ifpdf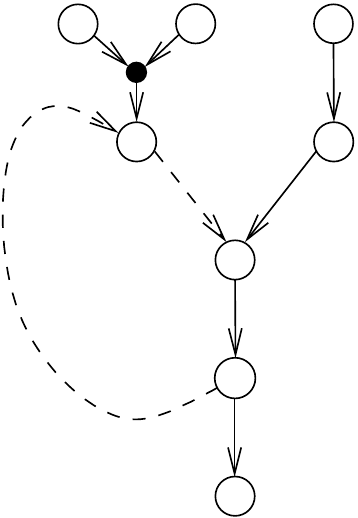\else\input{toynet.pstex_t}\fi}
  }%
  \hfil%
  \parbox{.33\textwidth}{%
    \begin{align*}
      A \land B &\liff D\\
      C &\liff E\\
      \lnot D\lor E &\liff F\\
      F &\liff G\\
      \lnot G &\liff D\\
      G &\liff H\\
    \end{align*}}
  \hfil

  \caption{A small signaling network from Example~\ref{ex:toynet}. The
    dashed lines denote inhibition while the black node means a logic and.}
  \label{fig:toynet}
\end{figure}

\begin{example}\label{ex:toynet}
The inequality description to the network shown in
Figure~\ref{fig:toynet} reads

\begin{equation}\label{eq:toynet}
  \begin{array}{rl@{\quad\quad}rl}
    (1-x_A)+(1-x_B)&\geq (1-x_D) & x_D+x_E&\geq x_F\\
    (1-x_A)&\leq (1-x_D) & x_D&\leq x_F\\
    (1-x_B)&\leq (1-x_D) &  x_E&\leq x_F\\
    x_C&=x_E & (1-x_G)&=x_D\\
    x_F&=x_G &  x_G&=x_H\\
    \multicolumn{4}{c}{x_l\in \{0,1\} \ \forall l}\\
  \end{array}
\end{equation}
Several scenarios can be tested with these inequalities. First of all
certain input and output patterns can be checked for validity. For
this purpose fix $x_A$, $x_B$, $x_C$ and $x_H$ to the desired value
and solve the IP with arbitrary objective value. If it is feasible,
the input/output pattern is a valid assignment. If one is interested in
the output of the network for a prescribed input pattern, one fixes
the inputs to interesting values again and solves the IP with the
objective to maximize $x_H$. The returned objective value is the
value of $H$. In the example the input pattern $x_A=0, x_B=1, x_C=0$
gives and output $x_H=1$.

Another issue for modeling signaling networks is to check the
completeness of the model. This can be done by checking whether the
set described by the inequalities~\eqref{eq:toynet} contains
feasible points. Our example is feasible, as the point $(x_A, x_B,
x_C, x_D, x_E, x_F, x_G, x_H)=(1,1,0,1,0,0,0,0)$ is valid.

Potential failure modes and corresponding suitable intervention
strategies can be found by testing knock-in/knock-out scenarios and
checking if this forces other variables to obtain a specific
value. Knock-in/knock-out scenarios are done by fixing various
variables to a desired value. To test whether other variables are
thereby fixed, we solve different IPs. Two IPs are needed for checking
if two variables have a certain value. In our example we set $x_E=0$
and $x_H=1$. To check whether e.g. $x_D$ and $x_A$ need to be fixed, we
solve the IP with the objective function $x_D+x_A$. The solution is 1
and the variables are $x_D=0$ and $x_A=1$. Thus, we know that $x_D$
must have the value 0. Another optimization problem with the objective
to minimize $x_A$ gives the solution 0, and thus $x_A$ is not
necessarily 1 but can have both values. 
\end{example}

\newpage


\section{Some Complexity Results for IFFSAT}
\label{sec:simplest-ip}
\label{sec:ip-formulations}

The \IFFSAT{} problem becomes easier if certain structural properties are
fulfilled, as has been shown in~\citep{haus-truemper-weismantel:07}.
From now on we will restrict the signaling networks to equivalence formulas with
only one literal on the right-hand-side, i.e. $|L_i^B|=1 \ \forall
i$, unless it is explicitly stated differently.

Initially we have to transform \IFFSAT{} to a special form that
allows us to perform the subsequent analysis.

\begin{definition}[cascade form]
  A signaling network~\eqref{eq:snsat} is called \emph{in cascade
    form} if for all clauses $S_i$ it holds that $|L^A_i|\leq 2$ and
  $|L^B_i|=1$.
\end{definition}

\begin{remark}
  Any signaling network can be transformed to cascade form by
  introducing additional
  literals and equivalence clauses. Indeed, this can be achieved by
  recursively applying the following replacement:
  $$S_i=\Bigg\{%
  \Big(\bigvee_{j\in J_i} A_j\Big)%
  \leftrightarrow%
   B%
  \Bigg\}$$ 
  gets replaced by
  $$
  S_i'=\{A_1\vee A_2\leftrightarrow C\}$$
  and 
  $$\textstyle S_i''=\left\{%
    \Big(\bigvee_{j\in J_i\setminus\{1,2\}} A_j\Big)\vee C%
    \leftrightarrow%
    B%
  \right\}$$
  where $C$ is a new literal.
\end{remark}

Secondly we will review some notation from logic.

\begin{definition}
  Let $S$ be a 3-\CNF{} formula. The undirected \emph{graph G(S)} is
  defined by the variables of $S$ as its nodes and for every
  $2$-clause of $S$ there is an edge between the corresponding nodes.
\end{definition}           

\begin{definition}[cutnode, cutnode condition]\hspace{2cm}
  \begin{enumerate}
  \item Let $a$, $b$ and $c$ be nodes of a graph. We call $c$ an \emph{a/b
      cutnode} 
    if removing $c$ from the graph disconnects the nodes $a$ and
    $b$.
  \item An \IFFSAT{} instance $S$ in cascade form fulfills the \emph{cutnode
      condition} if for every equivalence formula $S_i$, $B_i$ is an
    $A_{1i}/A_{2i}$ cutnode in $G(S)$.  
  \end{enumerate}
\end{definition}                              

After \IFFSAT{} is transformed to cascade form, the cutnode condition
can easily be checked by computing the connected components of $G(S)$.
In~\citep{haus-truemper-weismantel:07} it is proved that
\begin{theorem}
  If an instance of \IFFSAT{} in
  cascade form satisfies the cutnode condition, it can be solved in
  linear time.
\end{theorem} 

The cut\-node condition is
a restriction to realistic networks but the following example
shows that, e.g., feedback loops do not in general contradict the
condition.
 
\begin{example}\label{ex:cutnode}
  The network defined by 
  \begin{center}
  \begin{tabular}{l@{\quad}l}
    \multicolumn{1}{c}{formula} & \multicolumn{1}{c}{2-clause(s)}\\
    $\hphantom{\lnot}A\lor B\liff C$ & $\lnot A\lor C,\lnot B\lor C$\\
    $\hphantom{\lnot}A\to D$         & $\lnot A\lor D$\\
    $\hphantom{\lnot}D\to E$         & $\lnot D\lor E$\\
    $\lnot E\to A$                   & $\hphantom{\lnot}E\lor A$\\
  \end{tabular}
  \end{center}
\begin{figure}
  \centering
  \resizebox{!}{.15\textheight}{\ifpdf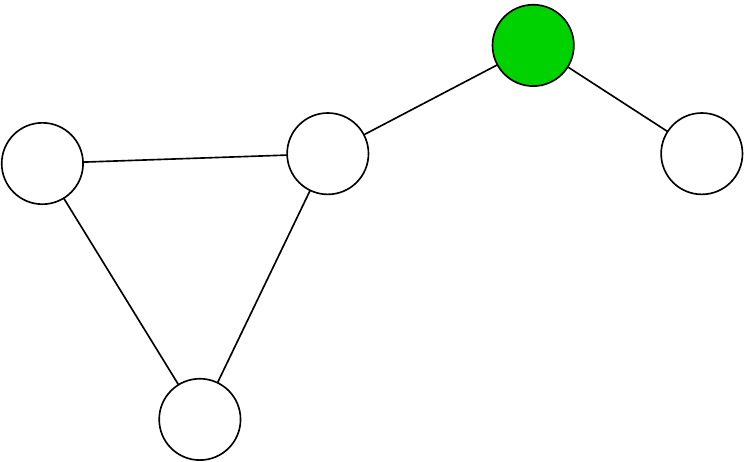\else\input{graph-g.pstex_t}\fi}
  \caption{Graph of the network from Example~\ref{ex:cutnode}.}
  \label{fig:cutnode}
\end{figure}

  shown in Figure~\ref{fig:cutnode} has a
  cycle with an odd number of negations, but satisfies the cutnode
  condition (there is one equivalence formula, and its output $C$ is a
  cutnode).
\end{example}

For the inequality description~\eqref{eq:sn-sat-ip} some nice
polyhedral properties can be obtained. One of them is
\emph{unimodularity}~\citep{bertsimas-weismantel:ip-book},
which leads to integral relaxations of the polyhedron, for the
smallest nonempty \IFFSAT{} problems.

\begin{definition}[unimodularity]
  A matrix $\m{A}\in \mathbf{Z}^{m\times n}$ of full row rank is
  \emph{unimodular} if the determinant of each basis of $\m{A}$ is
  $\pm 1$.
  \end{definition}

\begin{lemma}
  The submatrix of a single equivalence clause $S_i$ in the IP
  model~\eqref{eq:sn-sat-ip} is unimodular if the set $L^A_i$
  has cardinality $2$ and $L^B_i$ has cardinality $1$.
\end{lemma}

\begin{proof}
  Consider the case where the formula considered is exactly $A\vee B
  \leftrightarrow C$ with non-negated atoms $A,B,C$. The matrix
  $U=\scriptstyle\begin{pmat}
     1 &  1 & -1 \\
    -1 &  0 &  1 \\
     0 & -1 & -1 \\
   \end{pmat}$
   is clearly unimodular. Hence, $M=
   \begin{pmatrix}
     U & 0\\ 1& 1
   \end{pmatrix}$, where $1$ denotes a $3\times3$ unit matrix, is
   unimodular.

   All other cases arise from $M$ by unimodular row operations, i.e.
   subtracting the complementarity constraint in the top $3$ rows.
\end{proof}

Even for the simple formulas $(A\lor B) \liff (C\lor D)$ as well
as $(A\lor B\lor C) \liff D$, the inequality description~\eqref{eq:sn-sat-ip}
is non-unimodular. However, the linear relaxation
of~\eqref{eq:sn-sat-ip} for one equivalence
formula with arbitrary large $L_i^A$ and $L_i^B$ is still integral:
\begin{lemma}
  The inequality description of a single equivalence clause $S_i$ in
  the IP model~\eqref{eq:sn-sat-ip} is integral.
\end{lemma}
\begin{proof}
  See~\citep[p. 338]{hooker:07}.
\end{proof}

These integrality results cannot be generalized to \IFFSAT{} problems
with an arbitrary number of equivalence formulas, as the next example
shows. 

\begin{example}\label{ex:negfeed}
  We consider a set of two equivalences that model a small
  negative feedback cycle, i.e.

 $$S=\left\{(x_1\vee\bar{x}_3\leftrightarrow  x_4)
   \wedge (x_2\vee x_4\leftrightarrow x_3)\right\}.$$
  The integer programming formulation from~\eqref{eq:sn-sat-ip} has
  the following LP relaxation:

  \[\small
  \begin{array}{r@{}r@{}r@{}l@{\qquad}r@{}r@{}r@{}l}
    x_1&  + \bar{x}_3 &- x_4 &\geq 0&     x_2& - x_3&+ x_4
    &\geq 0\\
    -x_1&              &+ x_4 &\geq 0&   - x_2& + x_3&
    &\geq 0\\
    &  - \bar{x}_3 &+ x_4 &\geq 0&        & + x_3&- x_4
    &\geq 0\\
    \multicolumn{4}{c}{x_i+\bar{x}_i=1,} &\multicolumn{4}{c}{0\leq\;
      x_i\leq 1.}\\
  \end{array}\] 
  Computing the vertices we find both integral and fractional
  points: 
  \begin{center}
    {\ttfamily\small 
    \begin{tabular}{rrrrr} 
      $x_1$ & $x_2$ & $x_3$ & $\bar{x}_3$ &$x_4$\\ 
       0&   0& 1/2& 1/2& 1/2~  \\
       0& 1/2& 1/2& 1/2& 1/2~  \\   
      1/2&   0& 1/2& 1/2& 1/2~ \\ 
      1/2& 1/2& 1/2& 1/2& 1/2~ \\ 
       0&   1&   1&   0&   0~  \\ 
       1&   0&   1&   0&   1~  \\ 
       1&   1&   1&   0&   1. \\
     \end{tabular}}
  \end{center}
   It should be noted that not all vertices are fractional,
   which means that this \SAT{} instance is not infeasible.
 \end{example} 

\newpage
\section{Dynamics of Signaling Networks}
\label{sec:relaxed-sn}

In practical applications signaling networks often turn out to be
infeasible. This is not reasonable in a biological sense, but it can
occur due to delayed reactions modeled as instantaneous, or due to
modeling errors. Especially (negative) feedback loops with time
delays make the static models infeasible, but at the same time have a
huge impact on the functionality of a signaling network since
certain activation cascades can be enabled initially and switched off
at a later time point to avoid overreaction.
In order to model the dynamics of a signaling network we introduce an
extension of \IFFSAT{}, the \emph{requirement \IFFSAT{} problem}.
\begin{problem}[\RIFFSAT]
  Let a set $R$ of $r=|R|$ equivalence formulas of the form
  \begin{align}
    \label{eq:logic-y-model}
    \textstyle
    R=\left\{R_i=\bigvee_{j\in L^{A}_i} (A_j \land
    y_j)\liff ( B_i \land \bigvee_{j\in L^A_i} y_j )\right\}
  \end{align}
  over the literal set $L\cup Y$ be given, where $Y$ is the set of
  requirement variables. Let $L_0, L_1\subseteq L\cup Y$, sets of fixings,
  be given, then the \emph{\RIFFSAT} problem is to find a satisfying
  solution of
  \begin{equation}  
    \label{eq:riffsat}
    \bigwedge_{i=1}^r R_i \wedge \bigwedge_{x\in L_0} (\lnot x)\wedge
    \bigwedge_{x\in L_1} x.\tag{\RIFFSAT}
  \end{equation}
\end{problem}

In the signaling network context $y_j=1$ denotes influence of the
corresponding component $A_j$ on the right hand side while $y_j=0$
denotes no effect. In case $y_j=0$ for all $j\in L^A_i$, we request
$B_i$ to be free.

\begin{remark}
  We ask the requirement variable $y_j$ of $A_j$ to be the different for every
  \IFF{} formula in which $A_j$ occurs on the left hand side, and to
  be equal to $\bigvee_{k\in L^A_i} y_k$ if $A_j$ is the right hand
  side of clause $R_i$.
\end{remark}

We denote by $\F$ the set of all 0/1-points for
which~\eqref{eq:riffsat} is \TRUE,
i.e.\\ $\F=\left\{x\in\{0,1\}^{2|L|} \ | \
  \Gamma(\bigwedge_{i=1}^r R_i \wedge \bigwedge_{x\in L_0} (\lnot x)\wedge
  \bigwedge_{x\in L_1} x)=1 \right\}$. 

Usually one is interested in solving~\eqref{eq:riffsat} with special
properties on the set of requirement variables $y_j^i $, like a
maximal or minimal models over $y$. Such a problem can be solved by a
variation of \SAT, namely {\small MAXVAR~SAT} (see~\citep{truemper:04}). Here a satisfiable \CNF{} system $C$
and a set $T$ with \TRUE/\FALSE{} fixings of a variable subset is
given, such that $C$ is not satisfiable if all variables are fixed
according to $T$. The task is to determine a maximal subset $T^*$ of
$T$ so that $C$ is satisfiable. In our setting $T$ can be
$\{y_1=\ldots=y_{|L|}=\text{\sc True}\}$. {\small MAXCLS~SAT} is
also a related problem. We remark that both problems are special
cases of {\small MAXSAT}, which can not be approximated polynomially
better than $8/7$~\citep{hastad:01} and hence it is \NP\textit{-complete}.

Besides the presented method from logic one can also find a maximal solution
with the help of integer programming techniques.

\begin{lemma}[inequality description for \RIFFSAT]
  Given one equivalence formula $R_i\in R$ as in~\eqref{eq:logic-y-model},
  introduce $n_i=|L^A_i|$ additional binary variables
  $Y_i=\{x_{y_1},\ldots, x_{y_{n_i}}$. Then the corresponding set of
  feasible points $\F_{R_i}$ can be described by $n_i+n_i\cdot
  2^{n_i-1}$ inequalities plus binary constraints of the form:
  \begin{equation}
    \label{eq:ymodel}
    \begin{split}
      \begin{array}[t]{rl@{\quad}l}
        x_{B_i}&\geq x_{A_j}-(1-x_{y_j})&j\in L^A_i\\
        x_{B_i}&\leq \sum_{k\in S}{x_{A_k}}+(1-x_{y_j})+\sum_{k\notin
          S}{x_{y_k}}& j\in S, \ \emptyset\neq S\subseteq L^A_i \\
        x_l&\in \{0,1\}&\raisebox{-10pt}[0pt][0pt]{$l\in
          L_i^A\cup\{B_i\}\cup Y_i$}\\
        1&=x_l+\bar{x}_l& \\
      \end{array}
    \end{split}
  \end{equation}  
  We will assume that $\bar{x}_l$ is used in the inequality
  description if the corresponding atom is negated. To derive the
  inequality description for~\eqref{eq:riffsat}, introduce such
  inequalities for all $R_i$, $i=1,\ldots,r$, and fix the variables
  according to $L_0$ and $L_1$.
\end{lemma}

\begin{remark}
  Note that this formulation still preserves the form of a generalized
  set covering problem. In addition, for each fixed
  $x_y$ the formulation reduces to an instance of~\eqref{eq:sn-sat-ip}.
\end{remark}

Maximizing the sum over all requirement variables $x_{y_j}$ yields one
maximal feasible solution. However, there are many inequalities needed
to describe the feasible points. It turns out that all of them are
needed: 

\begin{lemma}\label{th:facets}
The inequalities for a single relaxed equivalence
clause~\eqref{eq:ymodel} are facets for the convex hull of
their integral points $\F$.
\end{lemma}

\begin{proof}
  We will show that for every inequality there exist
  $\dim(\mbox{conv}(\F))=2n+1$ affine independent, integral points
  fulfilling the corresponding inequality with equality, where $n=|L^A_i|$.

  $x_{B_i}\geq x_{A_j}-(1-y_j)$: For reasons of symmetry we restrict the analysis to
  $j=1$. The $2n+1$ linearly independent points are displayed in
  Table~\ref{tab:type1-points}.

\begin{table}
    \centering
    \begin{tabular}[htb]{ccccccccccc}
      \toprule
      $x_{A_1}$ & $x_{A_2}$ & $x_{A_3}$ & $\cdots$ & $x_{A_n}$& $x_{B_i}$ & $y_1$ & $y_2$ &
      $ y_3$ & $\cdots$ & $y_n$\\
      \midrule
      \multicolumn{11}{l}{$n$ points exploiting all free $y_j$:}\\
      1 & 0 & 0 & $\cdots$ & 0 & 1 & 1 & 1 & 1 & $\cdots$ & 1 \\
      1 & 0 & 0 & $\cdots$ & 0 & 1 & 1 & 0 & 1 & $\cdots$ & 1 \\
      $\vdots$ &$\vdots$&$\vdots$ &$\ddots$&$\vdots$&$\vdots$&$\vdots$&$\vdots$
      &$\vdots$& $\ddots$&$\vdots$\\
      1 & 0 & 0 & $\cdots$ & 0 & 1 & 1 & 1 & 1 & $\cdots$ & 0 \\
      \midrule
      \multicolumn{11}{l}{$n-1$ points exploiting all free $x_{A_j}$:}\\
      1 & 1 & 0 & $\cdots$ & 0 & 1 & 1 & 1 & 1 & $\cdots$ & 1 \\
      1 & 0 & 1 & $\cdots$ & 0 & 1 & 1 & 1 & 1 & $\cdots$ & 1 \\
      $\vdots$ &$\vdots$&$\vdots$ &$\ddots$&$\vdots$&$\vdots$&$\vdots$&$\vdots$
      &$\vdots$& $\ddots$&$\vdots$\\
      1 & 0 & 0 & $\cdots$ & 1 & 1 & 1 & 1 & 1 & $\cdots$ & 1 \\
      \midrule
      \multicolumn{11}{l}{two possible points with $x_{B_i}=0$:}\\
      1 & 0 & 0 & $\cdots$ & 0 & 0 & 0 & 1 & 1 & $\cdots$ & 1 \\
      1 & 1 & 0 & $\cdots$ & 0 & 0 & 0 & 1 & 1 & $\cdots$ & 1 \\
      \bottomrule
    \end{tabular}
    
\caption{$2n+1$ linearly independent points for $x_{B_i}\geq x_{A_1}-y_1$.}
\label{tab:type1-points}
\end{table}
  
  $x_{B_i} \leq \sum_{j\in S}{ x_{A_j}} + (1-y_l) + \sum_{j\notin S}{ y_j}$:\\
  For an easier notation, let us assume that the first $k$ indices
  belong to the selected set S and $l=1$. In
  Table~\ref{tab:type2-points} the linearly independent points are
  presented.

\begin{table}
    \centering
    \begin{tabular}[htb]{ccccccccccccccccc}
      \toprule
       $x_{A_1}$ & $x_{A_2}$ & $\cdots$ & $x_{A_k}$ & $x_{A_{k+1}}$& $x_{A_{k+2}}$ &
       $\cdots$ & $x_{A_n}$ & $x_{B_i}$ &$y_1$ & $y_2$ & $\cdots$ &
       $ y_k$ & $y_{k+1}$ & $y_{k+2}$ & $\cdots$ & $y_n$\\
       \midrule
       \multicolumn{17}{l}{$n-k+1$ points exploiting all free $x_{A_j}$:}\\
       0 & 0 & $\cdots$ & 0 & 0 & 0 & $\cdots$ & 0 & 0 & 1 & 1 &
       $\cdots$ & 1 & 0 & 0 & $\cdots$ & 0 \\
       0 & 0 & $\cdots$ & 0 & 1 & 0 & $\cdots$ & 0 & 0 & 1 & 1 &
       $\cdots$ & 1 & 0 & 0 & $\cdots$ & 0 \\
       0 & 0 & $\cdots$ & 0 & 0 & 1 & $\cdots$ & 0 & 0 & 1 & 1 &
       $\cdots$ & 1 & 0 & 0 & $\cdots$ & 0 \\
       $\vdots$ & $\vdots$ & $\ddots$ & $\vdots$ & $\vdots$ & $\vdots$ & $\ddots$ &
       $\vdots$ & $\vdots$ & $\vdots$ & $\vdots$ & $\ddots$ & $\vdots$
       & $\vdots$ &  $\vdots$ &  &$\vdots$ \\
       0 & 0 & $\cdots$ & 0 & 0 & 0 & $\cdots$ & 1 & 0 & 1 & 1 &
       $\cdots$ & 1 & 0 & 0 & $\cdots$ & 0 \\
       \midrule
       \multicolumn{17}{l}{$k-1$ points exploiting all free $y_j$:}\\
       0 & 0 & $\cdots$ & 0 & 0 & 0 & $\cdots$ & 0 & 0 & 1 & 0 &
       $\cdots$ & 1 & 0 & 0 & $\cdots$ & 0 \\
       $\vdots$ & $\vdots$ & $\ddots$ & $\vdots$ & $\vdots$ & $\vdots$ & $\ddots$ &
       $\vdots$ & $\vdots$ & $\vdots$ & $\vdots$ & $\ddots$ & $\vdots$
       & $\vdots$ &  $\vdots$ &  &$\vdots$ \\
        0 & 0 & $\cdots$ & 0 & 0 & 0 & $\cdots$ & 0 & 0 & 1 & 1 &
       $\cdots$ & 0 & 0 & 0 & $\cdots$ & 0 \\
       \midrule
       \multicolumn{17}{l}{$k$ points with one $x_{A_j}=1$, $j\in S$ each:}\\
       1 & 0 & $\cdots$ & 0 & 0 & 0 & $\cdots$ & 0 & 1 & 1 & 1 &
       $\cdots$ & 1 & 0 & 0 & $\cdots$ & 0 \\
       0 & 1 & $\cdots$ & 0 & 0 & 0 & $\cdots$ & 0 & 1 & 1 & 1 &
       $\cdots$ & 1 & 0 & 0 & $\cdots$ & 0 \\
       $\vdots$ & $\vdots$ & $\ddots$ & $\vdots$ & $\vdots$ & $\vdots$ & $\ddots$ &
       $\vdots$ & $\vdots$ & $\vdots$ & $\vdots$ & $\ddots$ & $\vdots$
       & $\vdots$ &  $\vdots$ &  &$\vdots$ \\
       0 & 0 & $\cdots$ & 1 & 0 & 0 & $\cdots$ & 0 & 1 & 1 & 1 &
       $\cdots$ & 1 & 0 & 0 & $\cdots$ & 0 \\
       \midrule
       \multicolumn{17}{l}{$n-k$ points with one $y_j=1$, $j\in
         L^A_i\setminus S$ each:}\\
       0 & 0 & $\cdots$ & 0 & 1 & 1 & $\cdots$ & 1 & 1 & 1 & 1 &
       $\cdots$ & 1 & 1 & 0 & $\cdots$ & 0 \\
       0 & 0 & $\cdots$ & 0 & 1 & 1 & $\cdots$ & 1 & 1 & 1 & 1 &
       $\cdots$ & 1 & 0 & 1 & $\cdots$ & 0 \\
       $\vdots$ & $\vdots$ & $\ddots$ & $\vdots$ & $\vdots$ & $\vdots$ & $\ddots$ &
       $\vdots$ & $\vdots$ & $\vdots$ & $\vdots$ & $\ddots$ & $\vdots$
       & $\vdots$ &  $\vdots$ &  &$\vdots$ \\
       0 & 0 & $\cdots$ & 0 & 1 & 1 & $\cdots$ & 1 & 1 & 1 & 1 &
       $\cdots$ & 1 & 0 & 0 & $\cdots$ & 1 \\
       \midrule
       \multicolumn{17}{l}{one possible point with $y_1=0$:}\\
       0 & 0 & $\cdots$ & 0 & 0 & 0 & $\cdots$ & 0 & 1 & 0 & 1 &
       $\cdots$ & 1 & 0 & 0 & $\cdots$ & 0 \\
       \bottomrule
     \end{tabular}

     \caption{$2n+1$ linearly independent points for $x_{B_i}\leq\sum_{j\in S}{
        x_{A_j}} + (1-y_1) + \sum_{j\notin S}{ y_j}$, with $S=\{1,\ldots,k\}$.}
    \label{tab:type2-points}
\end{table}
This concludes the proof.
\end{proof}

Lemma~\ref{th:facets} thus implies that the inequalities that are needed
to describe $\conv(\F)$ definitely include all constraints
in~\eqref{eq:ymodel}. Hence, the number of inequalities is exponential in the
number of inputs of the clause. This motivates to explore alternative,
extended formulations for modeling $\F$ based on a cascade
representation of their underlying \IFFSAT{} instance.
\begin{lemma}
  Each \RIFFSAT{} instance can be transformed to cascade form.
\end{lemma}
The recursive construction is analogous to the \IFFSAT{} case, handling
each conjunction as one literal and introducing variables of $y$-type
for the artificial variables.

Although transforming the \RIFFSAT{} instance to this specific form
leads to more variables and equivalence formulas, the growth of the
instance is only quadratic. Recall that $r=|R|$ and let $n$ be the
largest number of inputs of all clauses. Then, due to the
transformation, the number of literals and clauses in the \RIFFSAT{}
instance increases by at most $2r(n-2)$ and $r(n-2)$, respectively. In
contrast to that there are $6(r+r(n-2))$ inequalities needed to encode the
signaling network in cascade form compared to $(n+n\cdot 2^{n-1})r$
inequalities in the IP representation of Lemma~\ref{th:facets}. Both
counts leave out binary constraints and fixings. Hence, it is a
reasonable tool to reduce the complexity of the inequality description.\\

So far both approaches, {\small MAXVAR~SAT} as well as IP, provide one
maximal feasible solution. But one is interested in all maximal
solutions with respect to the requirement variables to find suitable
intervention strategies or to plan specific experiments verifying an
actual network structure. In order to find the set of all maximal
models (see Definition~\ref{def:model}) or all maximal feasible
solutions with respect to the y-variables it is necessary to use a
`clever enumeration method'. For this purpose we want to make use of
the joint generation algorithm~\citep{FK96}. This method checks
whether a pair of monotone boolean functions are dual. The problem is
equivalent to finding all maximal models of a monotone \CNF{}
formula.

\begin{definition}[monotone \CNF{}]
  Let $C$ be a \CNF{} expression and let $Z$ be a subset of the
  literals $L$. Then $C$ is called \emph{down-monotone}
  (\emph{up-monotone}) in $Z$ if from
  $m^*(Z,L\setminus Z)$ satisfying $C$ it follows that
  $m(Z,L\setminus Z)$ satisfies $C$ for all $m(Z)\leq m^*(Z)$
  ($m(Z)\geq m^*(Z)$) and $m(L\setminus Z)=m^*(L\setminus Z)$. 
\end{definition}
An integral set is down-monotone in the vector $z$ if $(x, z^*)\in P$ leads
to $(x, z)\in P$ for all $z\leq z^*$.

\begin{remark}
  A \CNF{} formula is down-monotone in $x$ if all literals
  of $x$ have only negative occurrence.
\end{remark}

Note that \RIFFSAT{} is not necessarily down-monotone in y as
the next example shows:

\begin{example}
  We consider the easy \RIFFSAT{} instance
  \begin{equation*}
   Z=(A\land y_A)\lor(B\land y_B)\liff C\land (y_A\lor y_B).
  \end{equation*}
  It is not down-monotone in $Y=\{y_A, y_B\}$ since the truth
  assignment
  \begin{equation*}
    A=1\quad B=0\quad  y_A=1\quad y_B=1\quad C=1
  \end{equation*}
  satisfies $Z$, while
  \begin{equation*}
    A=1\quad B=0\quad y_A=0\quad y_B=1\quad C=1
  \end{equation*}
  does not.
\end{example}

As \RIFFSAT{} instances are non-monotone in general, we apply a
transformation proposed \\by~\citep{kavvadias-sideri-stavropoulos:00}
to monotonize the requirement variables of non-monotone \CNF{}
terms. The method preserves the set of maximal models while the general
models can differ. 
Therefore the positive y-variables are eliminated according to a
recursive resolution procedure:
\begin{enumerate}
\item Expand the \RIFFSAT{} instance $K$ to \CNF{} form.
\item Choose a variable $y_k$ that occurs as positive and negative literal
  in $K$.
\item Divide the clauses into three parts $S_{y_k}\cup
  S_{\bar{y}_k} \cup A_{y_k}$, the set of clauses with occurrence of $y_k$,
  of $\bar{y}_k$, and no occurrence at all. In this context `$\cup$'
  denotes a conjunction. We write
  $K_i=(C_i\lor y_k)$ for $K_i\in S_{y_k}$, $i=1,\ldots,|S_{y_k}|$,
  and  $K_j=(D_j\lor \bar{y}_k)$ for $K_j\in S_{\bar{y}_k}$,
  $j=1,\ldots,|S_{\bar{y}_k}|$. Thus, $C_i$ and $D_j$ are
  disjunctions of all other variables apart from $y_k$.
\item Compute all \emph{resolvents} $R_{y_k}$ of each pair of clauses in
  $S_{y_k}$ and $ S_{\bar{y}_k}$ with respect to $y_k$,
  i.e. $R_{y_k}=\left\{ (C_i\lor D_j) \ | \ K_i\in S_{y_k}, \
    K_j\in S_{\bar{y}_k} \right\}$. 
\item The expression $K_{y_k}= R_{y_k}\cup S_{\bar{y}_k}\cup A_{y_k}$
  is monotone in $y_k$, since $y_k$ has only negative occurrences.
\end{enumerate}

\begin{theorem}[\citealt{kavvadias-sideri-stavropoulos:00}]
  With the above transformation it holds that
  $\mbox{models}(K)\subseteq\mbox{models}(K_{y_k})$ and
  $\maximal(K)=\maximal(K_{y_k})$.
\end{theorem}

Thus, we can apply this
concept recursively and obtain a \CNF{} expression which is monotone
in the y-variables and hence, we can utilize the joint generation
algorithm~\citep{FK96} to compute the invariant maximal models.

However, it is easy to see that the transformation can lead to an
exponentially larger expression. But in the case of \RIFFSAT{}
instances the number of clauses even decreases.

Monotonizing the requirement variables in one
$R_i\in R$, formula~\eqref{eq:logic-y-model} reduces the clauses to
\begin{align}
  \label{eq:mon-logic}
  \begin{split}
    & \textstyle\bigwedge_{j\in L_i^A} \left((\neg A_j\lor
      \neg y_j  \lor B_i)\land\bigvee_{k\in L_i^A} A_k
      \lor(\neg y_j \lor \neg B_i ) \right)\\
    \approx & \textstyle\left(\bigvee_{j\in L_i^A} (A_j\land y_j)\rightarrow B_i
    \right)\land \left(\bigvee_{j\in L_i^A} (C\land
      y_j)\rightarrow \bigvee_{k\in L_i^A} A_k\right).
  \end{split} 
\end{align}
In terms of inequalities this is modeled by
\begin{equation}\label{eq:mon-ineq}
  \begin{split}
    \begin{array}[t]{rl@{\quad}l}
      x_{B_i}&\geq x_{A_j}-(1-x_{y_j})&\raisebox{-10pt}[0pt][0pt]{$j\in L^A_i$}\\
      x_{B_i}&\leq \sum_{k\in L^A_i}{x_{A_k}}+(1-x_{y_j})&  \\
      x_l&\in \{0,1\}&\raisebox{-10pt}[0pt][0pt]{$l\in
        L_i^A\cup\{B_i\}\cup Y_i$.}\\
      1&=x_l+\bar{x}_l& \\ 
    \end{array}
  \end{split}
\end{equation}
Introducing such constraints for every $R_i\in R$ gives the
monotonized version of~\eqref{eq:riffsat}.  We denote by $\F_{mon}$
the set of integral points fulfilling~\eqref{eq:mon-logic} and thus
also~\eqref{eq:mon-ineq} for each $i=1,\ldots,r$.

\begin{lemma}
  The system~\eqref{eq:mon-logic}, and thus~\eqref{eq:mon-ineq}, are
  monotone in $y$ and have the same maximal y-solutions
  as~\eqref{eq:riffsat} and~\eqref{eq:ymodel}.
\end{lemma}

\begin{proof}
  From the inequalities it is easy to see that it is down-monotone in $y$,
  since pushing one y component to 0, say $x_{y_j}=0$, only relaxes both
  inequalities containing $x_{y_j}$. In particular, this means that
  the same $A$ and $B$ components are feasible for $x_{y_j}=0$ as for
  $x_{y_j}=1$. In the logic formula one can see the down-monotonicity
  in $y$ by expanding the implications to \CNF{}. All y-variables
  occur as negated atoms.

  To prove the accordance of the maximal solutions we can focus on a
  single \RIFFSAT{} formula $\bigvee_{i=1}^n (A_i\land y_i)\liff
  B\land \bigvee_{i=1}^n y_i$, $I=\{1,\ldots,n\}$, because the
  y variables are disjoint for
  every \RIFFSAT{} formula and thus the monotonized version of
  \RIFFSAT{} can be built by monotonizing each formula $R_i$.

  We will show that the following two integer programs~\eqref{eq:proof}
  and~\eqref{eq:proof-mon} have the same optimal solutions with
  respect to the $y$ components.
  \begin{equation}\label{eq:proof}
    \begin{split}
      \begin{array}[t]{crl@{\quad}l}
        \mbox{max}& \multicolumn{2}{l}{\sum_{i=1}^n{x_{y_i}}}\\ 
      \mbox{s.t.} &x_B&\geq x_{A_j}-(1-x_{y_j})&j\in I\\
        &x_B&\leq \sum_{k\in S}{x_{A_k}}+(1-x_{y_j})+\sum_{k\notin
          S}{x_{y_k}}& j\in S, \ \emptyset\neq S\subseteq I \\
        &x_l&\in \{0,1\}&\raisebox{-10pt}[0pt][0pt]{$l\in
          \{A_I, y_I, B\}$}\\
        &1&=x_l+\bar{x}_l& \\
      \end{array}
    \end{split}\tag{V}
  \end{equation}
  and 
  \begin{equation}\label{eq:proof-mon}
    \begin{split}
      \begin{array}[t]{crl@{\quad}l}
        \mbox{max}&\multicolumn{2}{l}{\sum_{i=1}^n{x_{y_i}}}\\
        \mbox{s.t.}& x_B&\geq
        x_{A_j}-(1-x_{y_j})&\raisebox{-10pt}[0pt][0pt]{$j\in I$}\\
        &x_B&\leq \sum_{k=1}^n{x_{A_k}}+(1-x_{y_j})&  \\
        &x_l&\in \{0,1\}&\raisebox{-10pt}[0pt][0pt]{$l\in
          \{A_I, y_I, B\}$}\\
        &1&=x_l+\bar{x}_l& \\ 
      \end{array}
    \end{split}\tag{W}
  \end{equation}
  For an ease of notation we will subsequently use the index $l$ for the
  variable $x_l$. Let $A_I$ and $y_I$ denote $(A_1,\ldots, A_n)$
  and $(y_1,\ldots, y_n)$. Let $\F_V$ and $\F_W$ be the set of
  integral points of problem~\eqref{eq:proof}
  and~\eqref{eq:proof-mon}, respectively, and let  $x^{*L}$ indicate
  the optimal solution of problem (L), L=V,W. Then, $y^{*L}$ denotes
  the corresponding $y$ components. Note that $\F_V\subset\F_W$. 

  {\bfseries Case 1:} If there are no $A_I$, $B$ components fixed to a
    certain value, there exists an $x^{*L}\in\F_L$ with
    $y^{*L}=\ve{1}$. For example, the vector $(A_I, y_I,
    B)=(\ve{0},\ve{1},0)$ is feasible for both IPs. Thus the optimal
    solutions are equal.

  {\bfseries Case 2:} If some $A_I$, $B$ components are fixed to 0 or
  1, there are two critical cases in which there exists no $x^L\in
  \F_L$ such that $y^L=1$ for L=V,W.\\
  The first one is $A_I=\ve{0}$ and $B=1$. Here, $y^{*L}=\ve{0}$ is
  necessary for both L since the second type of inequalities
  of~\eqref{eq:proof-mon}, which are also valid for~\eqref{eq:proof},
  forces $B$ immediately to 0 if one $y_j=1$. Thus, the optimal
  values match.\\
  The second critical configuration is that $A_i=1$ for $i\in
  I_1\subseteq I$ with $I_1\neq \emptyset$ and $A_j=0, \ j\notin I_1$,
  but $B=0$. The optimal y solution is $y_i^{*L}=0$ for $i\in I_1$ and $y_j^{*L}=1, \
  j\notin I_1$, L=V,W, by the first type of inequalities, which is
  valid for both programs. It forbids to lift another $y_i^L$ to 1 as
  this switches $B$ also to 1.

  If $x_{y_j}$ variables are fixed to 0 or 1, it influences both
  objective values equally. Thus, the objective value of the integer
  programs~\eqref{eq:proof} and~\eqref{eq:proof-mon} coincide in every
  possible case which concludes the proof.  
\end{proof}

\begin{lemma}\label{th:integrality-mon}
  The inequality system~\eqref{eq:mon-ineq} describes $\conv(\F_{mon})$
  when substituting $\{0,1\}$ by $[0,1]$.  
\end{lemma}

A proof for the Lemma can be found in the Appendix.
Note that the system~\eqref{eq:mon-ineq} are even a facet
description for $\conv(\F_{mon})$ as the
inequalities~\eqref{eq:ymodel} are facets for $\conv(\F)$ and all
feasible points are preserved. 

The LP relaxation of a monotonized dynamic IP~\eqref{eq:mon-ineq}
containing more than one formula, i.e. $r\geq 2$,  is not integral in
general as the following example shows.

\begin{example}
  Consider the monotonized version of the dynamic
  IP~\eqref{eq:mon-ineq} of the signaling network from
  Example~\ref{ex:negfeed}. Its linear relaxation is the following
  polyhedron
  \begin{equation*}
    \begin{array}{rl@{\quad\quad}rl}
      x_1-(1-y_1)&\leq x_4 &  x_2-(1-y_2)&\leq x_3\\
      (1-x_3)-(1-y_3)&\leq x_4 &  x_4-(1-y_4)&\leq x_3\\
      x_1+(1-x_3) + (1-y_1)&\geq x_4 & x_2+x_4+(1-y_2)&\geq x_3\\
      x_1+(1-x_3) + (1-y_3)&\geq x_4 &  x_2+x_4+(1-y_4)&\geq x_3\\
      \multicolumn{4}{c}{0\leq x_i, y_i\leq 1 \ \mbox{for } i=1,\ldots, 4}\\
    \end{array}
  \end{equation*}
  This polyhedron is not integral, since e.g. $(\ve{x}, \ve{y})=(1/2,
  1, 1/2, 1/2, 1, 1/2, 1, 1)$ is a vertex. But it also has integral
  vertices like $(0,0,1,0,1,0,1,0)$ so that it is not infeasible.
\end{example}

To illustrate the use of the dynamic model we give a small example of a
signaling network. 
\begin{figure}
  \hfil
  \parbox{0.65\textwidth}{
    \resizebox{!}{.3\textwidth}{\ifpdf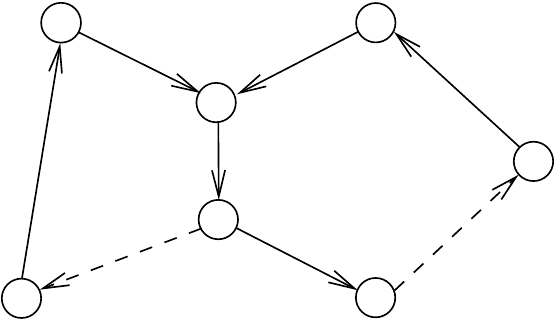\else\input{2-feedbacks.pstex_t}\fi}}
  \hfil
  \parbox{0.33\textwidth}{
    \begin{align*}
      A &\liff C\\
      B &\liff C\\
      C &\liff D\\
      \lnot D &\liff E\\
      E &\liff A\\
      D &\liff F\\
      \lnot F &\liff G\\
      G &\liff B
    \end{align*}}
  \hfil
    
  \caption{A signaling network with two negative feedback loops from
    Example~\ref{ex:feedbacks}. The dashed line denotes inhibition.}
  \label{fig:feedbacks}
\end{figure}

\begin{example}\label{ex:feedbacks}
Consider the network defined by
\begin{equation*}
  \begin{array}{rl@{\quad\quad}rl}
    x_C&\geq x_A-(1-y_A) & x_A&\geq x_E-(1-y_E)\\
    x_C&\leq x_A+(1-y_A) & x_A&\leq x_E+(1-y_E)\\
    x_C&\geq x_B-(1-y_B) & x_F&\geq x_D-(1-y_D)\\
    x_C&\leq x_B+(1-y_B) & x_F&\leq x_D+(1-y_D)\\
    x_D&\geq x_C-(1-y_C) & x_G&\geq (1-x_F)-(1-y_F)\\
    x_D&\leq x_C+(1-y_C) &  x_G&\leq (1-x_F)+(1-y_F)\\
    x_E&\geq (1-x_D)-(1-y_D) & x_B&\geq x_G-(1-y_G)\\
    x_E&\leq (1-x_D)+(1-y_D) & x_B&\leq x_G+(1-y_G)\\
    \multicolumn{4}{c}{x_l, y_l\in \{0,1\} \ \forall l}\\
  \end{array}
\end{equation*}
which is illustrated in Figure~\ref{fig:feedbacks}.
As there are two negative feedbacks involved, using the static
approach leads to infeasibility. Thus, we are interested in the
dynamics like possibly late reactions, i.e. in maximal feasible
subnetworks. In this case there are different possibilities to disturb
the cycles. One can either set $y_C=0$ and the rest to 1 or cut two
disjoint arcs each in one cycle, e.g. $y_E=0$ and $y_G=0$. The
remaining variables can then be 1. With this guideline and precise
experiments the practitioner can then identify how the network
structure actually looks like. Of course, the obtained y values also
give hints at suitable intervention strategies.
\end{example}

\begin{remark}
  Note that one can also model timing information using the
  \RIFFSAT{} formulation: Let $T\in Z_+$, a time horizon, be
  given. Then, for every $t\in \{1,\ldots,T\}$ we have different
  signaling networks, because of distinct reaction times. Due to these
  delays it is possible for a molecule to be absent at one point
  $t_1$, but to be present at $t_2$. This can be modeled by $T$
  copies of the same networks but each with a different $y^t$ vectors
  $t=1,\ldots,T$, which imply absence or presence of each
  interdependence in the biological unit at time point $t$. Each $y^t$
  vector is according to Equation~\eqref{eq:ymodel}. In this setting
  the changes of the signaling network over time is encoded by the
  difference of two consecutive $y^t$ vectors, $y^t- y^{t-1}$, $t\geq
  2$. 

  Our proposed approach is also able to handle extra information
  about the structure of the network over time. 
  Such information can be statements like if an interactions is
  present at time point $t^*$, it is present for all $t\geq t^*$, or
  an interaction is only valid for exactly one $t\in\{1,\ldots,T\}$.
  It can be expressed in terms of logical formulas/inequalities over the
  $y^t$. 
\end{remark}

\newpage
\section{Computations}

We have tested the ideas presented in this paper on
the \TCR-signaling network\\ from~\citep{immu-2}, as well
as several randomly generated signaling networks. To demonstrate
that the integer programming formulation is in fact useful to test
feasibility of scenarios (on the full model or in knock-in/knock-out
tests), we performed several tests: Table~\ref{tab:tcr-i/o} shows
feasibility tests for all possible combinations of input/output values
in the \TCR{} model. Computations were performed using \CPLEX~9.1 and
Allegro Common Lisp 8.1 on a \SUN-Fire-V890 with 1.2~GHz. Even though
the IPs do not generally reduce to an LP description, the whole
instances could be solved within two minutes.

\begin{table*}
  \centering
  \begin{tabular}{cccccccc}
    \toprule
    cols & rows & inputs & outputs & \#feas & \#infeas & total time (s) & avg time (s)\\
    \midrule
     214 & 376  &   3    &   14    &   36   &  131036  & $\approx 120$  & 0.001  \\
    \bottomrule
  \end{tabular}
  \caption{In/Out fixing of \TCR{} from~\citep{immu-2}.}
  \label{tab:tcr-i/o}
\end{table*}

Due to the lack of real instances we generated three types of random
signaling networks which are 
built with different probability distributions so that each has a
distinct structure. For technical reasons we deviate slightly from the
standard form of a signaling network: we construct \IFF formulas of the form
$\bigvee_{j\in J}A_j\to B$ (\OR-clauses) and $\bigwedge_{j\in J}A_j\to B$
(\AND-clauses). All types are constructed according to the following procedure:

\begin{enumerate}
\item Set the number of components $n$ and operation nodes $m$ and fix
  the proportion of \AND- and \OR-nodes $A/O$. In our random types $A/O$
  will always be equal to 0.5.
\item Determine, with respect to $A/O$, which operation nodes are \AND-
  and which are \OR-nodes. Therefore the realization of a random
  variable $X$ which probability distribution reflects the ratio $A/O$
  is generated. Clearly this means, $P(X=AND)=A/O=1-P(X=OR)$.
\item List all components from $1$ to $n$.
\item Let $a_i$ be the number of inputs of operation node
  $i=1,\ldots,m$. Determine $a_i$ as a realization of $A_i$ with
  $A_i\sim R$ with some probability distribution $R$. Note that $R$
  changes for the different generated types.
\item Choose the input and output components for each operation node
  separately by generating the realization $k$ of a random variable
  $K$, which is uniformly distributed on the set
  $\left\{1,\ldots,n\right\}$. The resulting number $k$ refers to the
  index of the components. 
  If a component is selected more than once, the duplicates are
  deleted leading to a smaller input degree. Note that for the
  selection of the output component of operation node $i$, the input
  components of $i$ must not be taken into account.
\item Decide if $S_k$ or $\overline{S}_k$ is assigned as an input,
  with a random number taking two equally probable values. 
\end{enumerate}

For the first type of random networks, we generate the number of
inputs according to a \\$\chi^2_f$~-~distribution. The parameter of this
distribution $f$, the degrees of freedom, equals the mean of
$A_i$. Hence, we set $A_i$ to be $\chi^2_3$- distributed with three
degrees of freedom in order to derive a `lean' signaling network
similar to the \TCR{} network with three inputs per operation node on
average. Results are displayed in Table~\ref{tab:type1}. The number of
molecules is denoted by `\# subs', `\# ops' denotes the number of
AND/OR operations used and `min/max/avg in' means the
minimal/maximal/average number of elements on the left hand side of all
\IFFSAT{} formulas $S_i$. The number of inputs and outputs to
the network are stated as `sources' and `sinks'. The size of the IP is
depicted in the following two rows. The IPs could be solved very fast
even in large cases and without any branch and bound nodes used. The number of
paths between sources and sinks is relatively large, there are, e.g., $14$,
$15$, $33$ and $21424$ paths. 

\begin{table}
  \centering
  \begin{tabular}{lrrrrrrrr}
    \toprule
             &10-5a&10-5b&20-10a&20-10b&100-50a&100-50b&100-100a&100-100b\\
\midrule
\# subs      &10   &10      &20     &20     &100    &100    &100   &100\\
\# ops       &5    &5       &10     &10     &50     &50     &100   &100\\
min in       &  1  &   1    & 1	    &	1   &1	    &   1   &1     & 1\\
max in       & 5   & 5      & 6	    &	7   &15	    &   9   &	11 & 12\\
avg in      & 3.2 & 3      & 3	    &2.6    &3.96   &   3.44&3.27  & 3.48\\
\# sources   &6    &   5    &	7   & 6     &38     &  30   &33    & 34 \\
\# sinks     & 1   & 1      & 2     & 3     & 6     &  5    & 3    & 2\\
\# variables &	26 &20      &  40   & 46    & 208   & 214   & 252  & 262 \\
\# rows      & 38  & 30     & 58    &  63   & 354   & 335   & 566  & 599 \\
time (sec)   &0.010&  0.000 & 0.000 & 0.000 & 0.010 & 0.010 & 0.010& 0.020 \\
\# B\&B nod  &  0  &   0    &   0   &    0  &   0   &   0   &  0   &0\\
feas/infeas  &  f  &   f    &   f   &   f   &   f   &   f   &  f   &f \\
\bottomrule
 \end{tabular}
 
 \caption{Overview of the properties of type I networks, which
   are specified by number of components -- number of operations. We
   always generated two networks, a and b, of the same size.}
 \label{tab:type1}
\end{table}

Additionally the shortest path between sinks and
sources has length $2$, even in the large networks. This is not a
realistic network structure. To avoid this feature we force the
network to hold more layers. This is done by generating several type I
networks, such that the sinks of the previous network are the sources
of the next one. Type II networks contain three Type I networks
glued together, which is reasonable as the shortest path in the \TCR{}
model is 5. Type II networks are similar to the real model with
respect to input degree number of paths and path length. Computations
are listed in Table~\ref{tab:type2} and are as fast as in the Type I
case. The number of paths from input to output layer decreases to $2$,
$1$, $11$, $422$ and $9243$, while its average length increases. Note
that the number of components and operations characterizing the
different networks differ from the entries in the table since they
specify the amount of components and operations in each Type I network.

\begin{table}
\centering
\begin{tabular}{lrrrrrrrr}
\toprule
            &5-1a   &5-1b  &10-5a &10-5b  &30-10a&30-10b &50-25a  &50-25b\\
\midrule
\# subs     &15     & 15   & 30   & 30    &90    &90     &150    &150\\
\# ops      & 3     & 3    & 15   & 15    & 30   &30     &75     &75 \\
min in      &1      &  1   & 0    & 1     &   0  &  0    &  0    & 0\\
max in      & 3     & 1	   & 6	  &  6	  &   5	 &  8    &8      & 9\\
avg in     & 1.67  & 1    & 1.6  &  1.87 &  2.13& 2.13  & 2.49  &	2.45\\
\# sources  & 2	    & 1	   &  8	  &  4	  &  12	 & 14    &  30   & 29\\
\# sinks    & 1	    &  1   &  4	  &  2    &  2   &  5	 &  10   & 13\\
\# variables& 30    &  30  &64	  & 74	  & 184	 & 184   & 318   & 322 \\
\# rows     & 22    & 22   & 74   & 70    & 186  & 186   & 428   & 429 \\
time (sec)  & 0.000 & 0.000& 0.000& 0.000 &0.010 & 0.010 & 0.010 & 0.010 \\
\# B\&B nod &0      &0     &0     &0      &0     &0      &0      &0 \\
feas/infeas &  f    &   f  &  f   &   f   &  f   &   f   &   f   & f  \\
\bottomrule			
\end{tabular}

\caption{Overview of the properties of type II networks, which
  are specified by number of components - number of operations.  We
   always generated two networks, a and b, of the same size.}
\label{tab:type2}
\end{table}

After the first two types were constructed similarly, the third type
is generated distinctly. We set $R:=U\left[1,15\right]$, i.e. $A_i$ is
uniformly distributed on the interval $\left[1,15\right]$. To obtain
integral input degrees, the generated random number is rounded up to
the next integer. This results in an average input number of
eight. Hence, the networks have comparably many arcs and are more or
less `thick'. The number of paths between sink and source are like in
Type I networks quite high, e.g. 24, 51, 26 494. In
Table~\ref{tab:type3} the computational results are displayed.

\begin{table}
\centering
\begin{tabular}{lrrrrrrrr}
  \toprule
             &25-5a  &25-5b  &50-10a &50-10b &100-20a&100-20b&150-50a&150-50b\\
\midrule
\# subs      &  25   &  25   &  50   &  50   &  100  &  100  & 150  & 150\\
\# ops       &  5    &  5    &  10   &  10   &  20   &  20   & 50   &50\\
min in       & 3     & 2     &	1    &   1   &  1    &  1    &   1  & 1\\
max in       & 9     & 10    &	13   &  10   &  14   &	13   &  14  & 14\\
avg in      & 	6.2  & 7     &	6.3  &4.8    &  7.2  &  7.6  & 7.24 & 7.1\\
\# sources   &	11   &	9    &	20   &	13   &	42   & 42    &  75  & 73\\
\# sinks     &	1    &	2    &	5    &	2    &   6   &  4    &  4   & 2\\
\# variables &	50   &	50   &	100  &	100  &	200  &204    & 306  &308\\
\# rows      &  65   &  61   &   123 &  108  & 263   & 277   & 568  & 561\\
time (sec)   &	0.000&  0.000& 0.000 & 0.010 & 0.000 & 0.010 & 0.010& 0.010\\
\# B\&B nod  &   0   &   0   &   0   &   0   &   0   &   0   &  0   &0\\
feas/infeas  &   f   &   f   &   f   &   f   &   f   &   f   &  f   &f \\
\bottomrule			
\end{tabular}

\caption{Overview of the properties of type III networks,
  which are specified by number of components -- number of operations. We
   always generated two networks, a and b, of the same size.}
\label{tab:type3}
\end{table}

Next, we employed the joint-generation method \citep{FK96} as
implemented in ~\citep{cl-jointgen} to compute all minimal
infeasible and maximal feasible vectors $y$. The feasibility oracle
employed in this algorithm is solving integer feasibility problems of
the form~\eqref{eq:ymodel}.


\begin{table*}
  \centering
  {\def~{\hphantom{.5}}
  \begin{tabular}{ccc@{\qquad}rrrr@{\qquad}rrrr@{\quad}rr}
    \toprule
    TCRB & CD4 & CD28 &
    \multicolumn{4}{c}{max feas}&
    \multicolumn{4}{c}{min infeas}&
    \# oracle calls&
    time\\
    \cmidrule{4-7}\cmidrule{8-11}
      &   &   & \# & max & avg & min & \# & max & avg & min\\
    \midrule
    0 & 0 & 0 &  1 & 324 &  324~ & 324 & 0  &     &       &     &    1 & 0\\
    0 & 0 & 1 &  1 & 322 &  322~ & 322 & 2  & 323 & 323~  & 323 &  974 & 0\\
    0 & 1 & 0 &  1 & 322 &  322~ & 322 & 2  & 323 & 323~  & 323 &  974 & 0\\
    0 & 1 & 1 &  1 & 320 &  320~ & 320 & 4  & 323 & 323~  & 323 & 1619 & 0\\
    1 & 0 & 0 &  2 & 323 & 322.5 & 322 & 2  & 322 & 322~  & 322 & 1298 & 0\\
    1 & 0 & 1 &  2 & 321 & 320.5 & 320 & 4  & 323 & 322.5 & 322 & 1943 & 0\\
    1 & 1 & 0 &  2 & 321 & 320.5 & 320 & 4  & 323 & 322.5 & 322 & 1943 & 0\\
    1 & 1 & 1 &  2 & 319 & 318.5 & 318 & 6  & 323 & 322.3 & 322 & 2584 & 0\\
    \bottomrule
  \end{tabular}
  }
  \caption{Computation of minimal infeasible and maximal feasible
    vectors $y$ on the TCR model
    from~\citep{immu-2}: Comparison of the cardinalities
    and computation time, and the support of the 
    $324$-dimensional vector of $y$-variables for all $8$ patterns of
    input signals. All timings below measurement threshold of $10$\,ms.}
  \label{tab:jg-y}
\end{table*}

\subsection*{Acknowledgments}
The first, second and fourth author were supprted by the FORSYS
initiative of the German Ministry of Education and Research through
the Magdeburg Center for Systems Biology (MaCS) and the Research Focus
Program Dynamic Systems funded by the Kultusministerium of
Saxony-Anhalt.


\bibliographystyle{jqt1999}

\bibliography{med,iba-bib,signalling,joint-generation}

\newpage

\appendix

\section{Appendix}

\begin{proof}[Proof of Lemma~\ref{th:integrality-mon}]
  Without loss of generality let $I=\{1,\ldots,n\}$ implying the
  \RIFFSAT{} formula
  $\bigvee_{i=1}^n (A_i\land y_i)\liff B\land \bigvee_{i=1}^n y_i$,
  $I=\{1,\ldots,n\}$. 
  Thus, let
  \begin{align*}
    \F_{mon}=\left\{ \left(
        \begin{array}{c}
          \ve{A} \\ \ve{y}\\ B\\
        \end{array}\right) \in \left\{0,1\right\}^{2n+1}: \right. 
    B & \left.\geq A_i-(1-y_i) \hspace{0.5cm}\forall \ i=1,\ldots,n \right. \\
    B &\leq\left. \sum_{j=1}^n{ A_j} + (1-y_i) \hspace{0.5cm} \forall
      \ i=1,\ldots,n \right\}\\ 
  \end{align*}
  and
  \begin{align*}
    P=\left\{ \left(
        \begin{array}{c}
          \ve{A} \\\ve{y}\\ B\\
        \end{array}\right) \in \left[0,1\right]^{2n+1}:\right. 
    B & \left. \geq A_i-(1-y_i) \hspace{0.5cm}\forall \ i=1,\ldots,n \right. \\
    B & \leq\left. \sum_{j=1}^n{ A_j} + (1-y_i) \hspace{0.5cm}
      \forall \ i=1,\ldots,n \right\}.\\
  \end{align*}
  We need to show that $P=\conv(\F_{mon})$. It is clear that
  $\conv(\F_{mon})\subseteq P$. For the converse, assume that
  $P$ has a fractional vertex $\ve{f}=(\ve{A}, \ve{y}, B)$. Since for a
  vertex $2n+1$ inequalities have to be tight, $\ve{f}$ has at least one
  integral component (there are only $2n$ inequalities without box
  conditions). Let $n_f$ be the number of fractional components of
  $\ve{f}$; $0\leq n_f \leq 2n$. It is to show that $n_f>0$ gives a
  contradiction.
  
  Let $k_1$ be the number of tight inequalities of the first type and
  $k_2$ the one of type two inequalities, $k_1+k_2=n_f$. $I_j$ denotes
  the corresponding equality index set of type $j=1,2$. To show the
  Lemma we will distinguish between
  three cases:
  \begin{enumerate}
  \item[(i)] $k_1>0 \land k_2>0$,
  \item[(ii)] $k_1=n_f>0 \land k_2=0$,
  \item[(iii)] $k_1=0 \land k_2=n_f>0$.
  \end{enumerate}

  {\bfseries Case (i):} Let $I_1, I_2\neq\emptyset$.\\
  Since $B=\sum_{j=1}^n{A_j} + (1-y_l)$ for all $l\in I_2$,
  $y_l$ must have the same value for all $l\in I_2$. Additionally, for
  $i\in I_1$ and $l\in I_2$ it holds that 
   \begin{align*}
     \textstyle
    A_i-(1-y_i)=\sum_{j=1}^n{A_j} + (1-y_l)\quad \Leftrightarrow
    \quad 0=\sum_{j\neq i}{A_j} + (1-y_i) + (1-y_l).
  \end{align*}
  Since all variables are non-negative, we find that $A_j=0 \ \forall
  j\neq i$, $y_i=y_l=1 \ \forall l\in I_2$ and $A_i=B$. Thus,
  $A_{I\setminus \{i\}}=\ve{0}$, $y_{I_1\cup I_2}=\ve{1}$ and the
  only possible fractional components of $\ve{f}$ are $y_j$ for $j\in
  I\setminus(I_1\cup I_2)=\{l_1,\ldots, l_k\}$ and $B=A_i$. But if
  they are fractional, we can construct a convex combination of
  integral points in $\F_{mon}$, because the variables are all strictly
  between 0 and 1 and all possible 0/1 combinations of the fractional
  variables are feasible points. For this purpose let w.l.o.g. the
  fractional $y_{I\setminus(I_1\cup I_2)}$ be ordered so that
  $y_{l_1}\leq\ldots\leq y_{l_m} \leq B\leq y_{l_{m+1}}\leq\ldots\leq
  y_{l_k}$, and let $\ve{v}_S^A=\sum_{i\in S} \ve{e}_i^A, \ S\subseteq
  I$, where $\ve{e}_i^A$ is the i-th unity vector on the $A$-variables
  and the remaining n+1 components are 0. $\ve{v}_S^y$ is analogously
  defined; $\ve{e}_B\in\{0,1\}^{2n+1}$ is the unity vector of $B$. Then
  the convex combination for $\ve{f}$ is 
  \begin{align*}
    \ve{f}=&y_{l_1}( \ve{e}^A_i + \ve{v}^y_I+\ve{e}_B) +
    (y_{l_2}-y_{l_1})(\ve{e}^A_i + \ve{v}^y_{I\setminus \{l_1\}} +
      \ve{e}_B)+\cdots +(B-y_{l_m}) \\
    &( \ve{e}^A_i+\ve{v}^y_{I\setminus\{l_1,\ldots,l_m\}}+\ve{e}_B)
     + (y_{l_{m+1}}-B) \ \ve{v}^y_{I\setminus
        \{l_1,\ldots,l_m\}}+\cdots+ (1-y_{l_k}) \ \ve{v}_{I_1\cup I_2}^y.
  \end{align*}
  If several components have the same value, this construction
  reduces them to one convex multiplicator, i.e. in the next vector
  they will all occur as 0. Thus, $\ve{f}$ is no vertex of
  $P$. {\Large \Lightning}\\

  {\bfseries Case (ii):} Let $I_1=\{1,\ldots,k_1\}$ and
  $I_2=\emptyset$.\\
  The number of integral components of $\ve{f}$ must be greater or
  equal to $n+1$, because $n_f=k_1\leq n$. Thus, there exists an
  index $j\in I_1$ such that the equality $B=A_j-(1-y_j)$ contains at
  least two integral variables. Thus, all contained variables are
  integral and in particular $B$ is integral. In the equalities
  where only $B$ is known to be integral, say $B=A_k-(1-y_k)$, $k\in
  I_1$, $A_k$ and $y_k$ can be fractional. There are two cases to be
  considered: $B=1$ and $B=0$.\\
  {\bfseries B=1:} $\ \Rightarrow 1=A_i-(1-y_i)\ \forall i\in I_1
  \quad \Leftrightarrow\quad 2=A_i+y_i\quad\Rightarrow
  A_i=y_i=1$. Thus, the only fractional components of $\ve{f}$ can be
  $A_{I\setminus I_1}$ and $y_{I\setminus I_1}$. As $B=1$ \emph{and}
  $A_i=y_i=1$ for all $i\in I_1$, there is a `reason' for $B$ being
  1. Therefore the fractional components can be convexly combined as
  in case (i) since $A_{I\setminus I_1}$ and $y_{I\setminus I_1}$ can
  be any 0/1 combination to complete a feasible point for $P$. Hence,
  $\ve{f}$ is no vertex. {\Large \Lightning}.\\
  {\bfseries B=0:} $\ \Rightarrow A_i=1-y_i \ \forall i\in I_1$. Let
  $\{p_1,\ldots,p_m\}=I_1^f\subseteq I_1$ denote the index set of the fractional
  $A$-variables in $I_1$ and thus their corresponding $y$-variables, which also
  have to be fractional according to the equality. For $j\in
  I\setminus I_1$ the inequalities imply that $A_j+y_j<1$. We proceed
  similar to the previous case. W.l.o.g. assume that $A_i\leq y_i \ \forall
  i\in I_1^f$ and that $A_{p_1}\leq y_{j_1}\leq A_{p_2}\leq\ldots\leq
  A_{j_k}\leq y_{j_l}\leq A_{p_m}$, with $j_i\in I\setminus I_1,
  i=1,\ldots, l$,  are the ordered fractional components. For ease of
  notation we will leave out the integral $A_i$ and $y_i$, $i\in
  I_1\setminus I_1^f$ and consider only the remaining components 
  $\ve{f}^*$. Accordingly, $\ve{v}_S^A, \ve{v}_S^y$ and $\ve{e}_B$ are
  adjusted to the new dimension. Then, 
  \begin{align*}
    \ve{f}^*=&A_{p_1}(\ve{v}^A_{I_1^f\cup I\setminus I_1}+\ve{v}^y_{I\setminus I_1}) +
    (y_{j_1}-A_{p_1})(\ve{v}^A_{(I_1^f\cup I\setminus
        I_1)\setminus\{p_1\}}+\ve{v}^y_{I\setminus I_1\cup \{p_1\}})
      +(A_{p_2}-y_{p_1})(\ve{v}^A_{(I_1^f\cup I\setminus
        I_1)\setminus\{p_1\}}+\\ &
      \ve{v}^y_{I\setminus I_1\cup\{p_1\}\setminus \{j_1\}})
      +\cdots + (A_{j_k}-A_{p_{m-1}})(\ve{e}^A_{j_k} + \ve{v}^y_{I_1^f\cup \{j_l\}})+
      (y_{j_l}-A_{j_k})\ve{v}^y_{I_1^f\cup \{j_l\}}  + (y_{p_1}-y_{p_2})\ve{v}^y_{I_1^f}.
  \end{align*}
  The vectors used in this representation are all elements of
  $P$. Furthermore the convex multiplicator sum up to 1, as
  $A_i+y_i=1$ for $i\in I_1^f$. Thus, $\ve{f}$ is not a
  vertex. {\Large \Lightning}\\ 

  {\bfseries Case (iii):} Let $I_2=\{1,\ldots,k_2\}$ and
  $I_1=\emptyset$.\\
  Analogous to case (i), it follows that $y_1=\ldots=y_{k_2}$. We
  also know that there are $n_f=k_2$ fractional components of
  $\ve{f}$.\\
  If these fractional components are $y_i$, $i\in I_2$, the
  other variables are all integral. But the equality
  $B=\sum_{j=1}^n A_j + (1-y_i)$ for all $i\in I_2$ and the
  integrality of all involved variables but $y_i$ yield that
  $y_i\in\{0,1\}$ $\forall i\in I_2$ and therefore $\ve{f}$ is
  integral. \\
  Thus, $y_i$ must be integral for all $i\in I_2$. We need
  to distinguish between 6 cases:
  \begin{enumerate}
  \item[a)] $A_i$ for $i\in I^A\subseteq I$ is fractional ($|I^A|=k_2$), 
  \item[b)] $y_i$ for $i\in I^y\subseteq I\setminus I_2$ is fractional
    ($|I^y|=k_2$), 
  \item[c)] $B, A_i$, $i\in I^A\subseteq I$ are fractional
    ($|I^A|=k_2-1$), 
  \item[d)] $B, y_i$, $i\in I^y\subseteq I\setminus I_2$ are fractional
    ($|I^y|=k_2-1$), 
  \item[e)] $A_i, y_j$, $i\in I^A\subseteq I$, $j\in I^y\subseteq
    I\setminus I_2$ are fractional ($|I^A|+|I^y|=k_2$),
  \item[f)]  $B, A_i, y_j$, $i\in I^A\subseteq I$, $j\in I^y\subseteq
    I\setminus I_2$ are fractional ($|I^A|+|I^y|=k_2-1$),
  \end{enumerate}
  The case where only $B$ is fractional does not have to be
  considered as equality in at least one type two inequality
  immediately forces $B$ to be integral, as well.

  a) With the equality it also follows that $\sum_{i\in I^A} A_i\in
  \{0,1\}$.
  \begin{itemize}
  \item $\sum_{i\in I^A} A_i=0 \ \Rightarrow \ A_i=0 \ \forall i\in
    I^A$. \Lightning
  \item $\sum_{i\in I^A} A_i=1$. $B=\sum_{j=1}^n A_j + (1-y_l)$, $l\in
    I_2 \ \Rightarrow \ B=1,\, A_i=0,\, i\in I\setminus I^A,\, y_l=1,\, l\in
    I_2$. As $B=1$, we only need one $A_i$, $i\in I^A$ to be 1,
    independent of the $y$ pattern, for a feasible point. Let
    $I^A=\{i_1,\ldots,i_m\}$, then $\ve{f}$ can be convexly combined by
    \begin{equation*}
      \ve{f}=A_{i_1}(\ve{e}_{i_1}^A + \ve{v}^y_{I_2} + \ve{v}^y_S
      +\ve{e}_B) +\cdots + A_{i_m}(\ve{e}_{i_m}^A + \ve{v}^y_{I_2} +
      \ve{v}^y_S+\ve{e}_B)
    \end{equation*}
    for the subset $S$ of $I\setminus I_2$, where $y_j=1$. Since we
    know that $\sum_{i\in I^A} A_i=1$, we have a real convex
    combination. \Lightning
  \end{itemize}
  
  b) As $I^y\cap I_2=\emptyset$, $\sum_{j=1}^n A_j +
  (1-y_k)=B<\sum_{j=1}^n A_j + (1-y_i)$, $i\in I^y$ $\Rightarrow$
  $y_i<y_k$ for all $k\in I_2$.
  \begin{itemize}
  \item $y_k=0$ for $k\in I_2 \ \Rightarrow y_i<0$ for $i\in
    I^y$. \Lightning 
  \item $y_k=1$. Thus, either $B=0=A_1=\ldots=A_n$ or (w.l.o.g.)
    $B=1=A_1, \, A_2=\ldots=A_n=0$. In both cases, we can choose any
    $y_{I^y}$ pattern to get a feasible point in $P$. Therefore, one
    can follow the same procedure to construct a convex combination
    of integral points as in case (i). \Lightning
  \end{itemize}

  c) Again, equality implies $B-\sum_{i\in I^A} A_i\in\{0,1\}$.
  \begin{itemize}
  \item $B-\sum_{i\in I^A} A_i=1 \ \Rightarrow \ B=1, \,A_i=0$, $i\in
    I^A$. \Lightning
  \item $B-\sum_{i\in I^A} A_i=0 \ \Leftrightarrow \ 0=\sum_{j\notin
      I^A} A_j + (1-y_k) \ \forall k\in I_2$\\
    $\Rightarrow A_j=0, \, j\notin I^A$ and $y_k=1, \, k\in
    I_2$. Additionally, $y_j\in\{0,1\}$ for $j\in I\setminus I_2$ and
    for a subset $S\subseteq I\setminus I_2$ $y_j=1$. Then we can
    convexly combine integral elements of $P$ to obtain $\ve{f}$. For
    this, let $I^A=\{i_1,\ldots, i_m\}$:
    \begin{equation*}\textstyle
      \ve{f}=A_{i_1}(\ve{e}^A_{i_1}+ \ve{v}^y_{I_2}+ \ve{v}^y_S
      +\ve{e}_B) + \cdots + A_{i_m}(\ve{e}^A_{i_m}+ \ve{v}^y_{I_2}+
      \ve{v}^y_S +\ve{e}_B) + (1-\sum_{i\in I^A} A_i)(\ve{v}^y_{I_2}+
      \ve{v}^y_S)
    \end{equation*}
    Therefore, $\ve{f}$ is no vertex. \Lightning
  \end{itemize}
 
  d) Since $B=\sum_{j=1}^n A_j + (1-y_k)$, $k\in I_2$ and all
  variables involved except for $B$ are integral, $B$ must be
  integral as well. Hence, we are in case b). \Lightning 

  e) As shown in case b) it follows that $y_i<y_k$ for all $i\in I^y$
  and $k\in I_2$ and with this $y_k=1 \,\forall k\in I_2$. This implies
  that $B=\sum_{j=1}^n A_j$ and hence $\sum_{i\in I^A} A_i\in\{0,1\}$.
  \begin{itemize}
  \item $\sum_{i\in I^A} A_i=0 \ \Rightarrow \ A_i=0 \, \forall i\in
    I^A$. The only remaining fractional components are $y_j$, $j\in
    I^y$. The possible cases for the integral components are analogous
    to case b) and they can be convexly combined as in case
    (i). \Lightning
  \item $\sum_{i\in I^A} A_i=1 \ \Rightarrow \  A_i=0 \, i\notin I^A,
    \, B=1,\, y_j=1, \, j\in I_2$. We again want to apply the scheme
    used in case (i), i.e. look for the minimum within the fractional
    components, set all fractional components to 1 in the convex
    combinator, multiply it by the lowest value, then take the
    minimal difference between the fractional entries and the
    previous minimum, set the component of the previous minimum to 0
    and multiply it by the difference of the two. For this purpose it
    is necessary to make sure that there is always at least one $A_i$,
    $i\in I^A$, that is one in the convex combinator. This can be
    achieved by adjusting the system a little bit to a combination of
    the procedures of case (i) and case (iii), a): We set only one
    $A_i$ at a time to 1 until its value is reached and then the next
    becomes 1. For the $y$ components an arbitrary pattern can be
    chosen. For an easier description assume that $I^A=\{i_1,\ldots,
    i_m\}, I^y=\{j_1,\ldots, j_k\}$ and $A_{i_1}\leq y_{j_1}\leq
    A_{i_2}\leq\ldots$:
    \begin{align*}
      \ve{f}=&A_{i_1}(\ve{e}_{i_1}^A+\ve{v}^y_{I_2}+\ve{v}^y_{I^y}
      +\ve{e}_B) + (y_{j_1}-A_{i_1})(\ve{e}_{i_2}^A+\ve{v}^y_{I_2}+\ve{v}^y_{I^y}
      +\ve{e}_B)\\ &+ (A_{i_2}-y_{j_1})(\ve{e}_{i_2}^A+\ve{v}^y_{I_2}
      +\ve{v}^y_{I^y\setminus\{j_1\}} +\ve{e}_B)+\cdots
    \end{align*}
    Since $\sum_{i\in I^A} A_i=1$ it follows that the constructions
    yields a real convex combination. \Lightning
  \end{itemize}

  f) In this case one can construct a convex combination of $\ve{f}$
  by using a combination of the cases e) and c). {\Large \Lightning}

\mbox{ }\\
In a all cases the fractional vertex $\ve{f}$ could be convexly
combined and thus can not be a vertex. This implies $P=\conv(\F_{mon})$
\end{proof}

\end{document}

%% file: toynet.pdftex_t
\begin{picture}(0,0)%
\includegraphics{toynet.pdf}%
\end{picture}%
\setlength{\unitlength}{4144sp}%
\begingroup\makeatletter\ifx\SetFigFontNFSS\undefined%
\gdef\SetFigFontNFSS#1#2#3#4#5{%
  \reset@font\fontsize{#1}{#2pt}%
  \fontfamily{#3}\fontseries{#4}\fontshape{#5}%
  \selectfont}%
\fi\endgroup%
\begin{picture}(1623,2354)(97,-1868)
\put(1171,-1816){\makebox(0,0)[b]{\smash{{\SetFigFontNFSS{8}{9.6}{\rmdefault}{\mddefault}{\updefault}{\color[rgb]{0,0,0}$H$}%
}}}}
\put(451,344){\makebox(0,0)[b]{\smash{{\SetFigFontNFSS{8}{9.6}{\rmdefault}{\mddefault}{\updefault}{\color[rgb]{0,0,0}$A$}%
}}}}
\put(991,344){\makebox(0,0)[b]{\smash{{\SetFigFontNFSS{8}{9.6}{\rmdefault}{\mddefault}{\updefault}{\color[rgb]{0,0,0}$B$}%
}}}}
\put(1621,344){\makebox(0,0)[b]{\smash{{\SetFigFontNFSS{8}{9.6}{\rmdefault}{\mddefault}{\updefault}{\color[rgb]{0,0,0}$C$}%
}}}}
\put(1621,-196){\makebox(0,0)[b]{\smash{{\SetFigFontNFSS{8}{9.6}{\rmdefault}{\mddefault}{\updefault}{\color[rgb]{0,0,0}$E$}%
}}}}
\put(721,-196){\makebox(0,0)[b]{\smash{{\SetFigFontNFSS{8}{9.6}{\rmdefault}{\mddefault}{\updefault}{\color[rgb]{0,0,0}$D$}%
}}}}
\put(1171,-736){\makebox(0,0)[b]{\smash{{\SetFigFontNFSS{8}{9.6}{\rmdefault}{\mddefault}{\updefault}{\color[rgb]{0,0,0}$F$}%
}}}}
\put(1171,-1276){\makebox(0,0)[b]{\smash{{\SetFigFontNFSS{8}{9.6}{\rmdefault}{\mddefault}{\updefault}{\color[rgb]{0,0,0}$G$}%
}}}}
\end{picture}%

%% file: graph-g.pdftex_t
\begin{picture}(0,0)%
\includegraphics{graph-g.pdf}%
\end{picture}%
\setlength{\unitlength}{4144sp}%
\begingroup\makeatletter\ifx\SetFigFontNFSS\undefined%
\gdef\SetFigFontNFSS#1#2#3#4#5{%
  \reset@font\fontsize{#1}{#2pt}%
  \fontfamily{#3}\fontseries{#4}\fontshape{#5}%
  \selectfont}%
\fi\endgroup%
\begin{picture}(3403,2096)(3857,-4034)
\put(5356,-2671){\makebox(0,0)[b]{\smash{{\SetFigFontNFSS{10}{12.0}{\familydefault}{\mddefault}{\updefault}{\color[rgb]{0,0,0}$A$}%
}}}}
\put(4771,-3886){\makebox(0,0)[b]{\smash{{\SetFigFontNFSS{10}{12.0}{\familydefault}{\mddefault}{\updefault}{\color[rgb]{0,0,0}$D$}%
}}}}
\put(4051,-2716){\makebox(0,0)[b]{\smash{{\SetFigFontNFSS{10}{12.0}{\familydefault}{\mddefault}{\updefault}{\color[rgb]{0,0,0}$F$}%
}}}}
\put(7066,-2671){\makebox(0,0)[b]{\smash{{\SetFigFontNFSS{10}{12.0}{\familydefault}{\mddefault}{\updefault}{\color[rgb]{0,0,0}$B$}%
}}}}
\put(6301,-2176){\makebox(0,0)[b]{\smash{{\SetFigFontNFSS{10}{12.0}{\familydefault}{\mddefault}{\updefault}{\color[rgb]{0,0,0}$C$}%
}}}}
\end{picture}%

%% file: 2-feedbacks.pdftex_t
\begin{picture}(0,0)%
\includegraphics{2-feedbacks.pdf}%
\end{picture}%
\setlength{\unitlength}{4144sp}%
\begingroup\makeatletter\ifx\SetFigFontNFSS\undefined%
\gdef\SetFigFontNFSS#1#2#3#4#5{%
  \reset@font\fontsize{#1}{#2pt}%
  \fontfamily{#3}\fontseries{#4}\fontshape{#5}%
  \selectfont}%
\fi\endgroup%
\begin{picture}(2537,1456)(533,-968)
\put(811,344){\makebox(0,0)[b]{\smash{{\SetFigFontNFSS{7}{8.4}{\rmdefault}{\mddefault}{\updefault}{\color[rgb]{0,0,0}$A$}%
}}}}
\put(2251,344){\makebox(0,0)[b]{\smash{{\SetFigFontNFSS{7}{8.4}{\rmdefault}{\mddefault}{\updefault}{\color[rgb]{0,0,0}$B$}%
}}}}
\put(2971,-286){\makebox(0,0)[b]{\smash{{\SetFigFontNFSS{7}{8.4}{\rmdefault}{\mddefault}{\updefault}{\color[rgb]{0,0,0}$G$}%
}}}}
\put(1531,-16){\makebox(0,0)[b]{\smash{{\SetFigFontNFSS{7}{8.4}{\rmdefault}{\mddefault}{\updefault}{\color[rgb]{0,0,0}$C$}%
}}}}
\put(1531,-556){\makebox(0,0)[b]{\smash{{\SetFigFontNFSS{7}{8.4}{\rmdefault}{\mddefault}{\updefault}{\color[rgb]{0,0,0}$D$}%
}}}}
\put(2251,-916){\makebox(0,0)[b]{\smash{{\SetFigFontNFSS{7}{8.4}{\rmdefault}{\mddefault}{\updefault}{\color[rgb]{0,0,0}$F$}%
}}}}
\put(631,-916){\makebox(0,0)[b]{\smash{{\SetFigFontNFSS{7}{8.4}{\rmdefault}{\mddefault}{\updefault}{\color[rgb]{0,0,0}$E$}%
}}}}
\end{picture}%